\documentclass[11pt,titlepage]{amsart}

\usepackage{atmcs} 
\usepackage{lipsum,enumerate,hyperref}
\usepackage[backend=biber,style=numeric,sorting=none,maxnames=99,minnames=99,url=true]{biblatex}

\renewbibmacro*{url+urldate}{%
  \printfield{url}
}
\usepackage{mathtools,tikz}
\usetikzlibrary{arrows,decorations.pathmorphing,backgrounds,calc,math}
\tikzstyle{map}=[->,semithick]
\tikzstyle{arc}=[bend left,->,semithick]
\tikzstyle{rinclusion}=[right hook->,semithick]
\tikzstyle{linclusion}=[left hook->,semithick]

\def\Ri{\mathcal R}

\def\R{{\mathbb R}}

\def\C{\mathcal C}

\def\eps{{\varepsilon}}

\newcommand{\norm}[1]{\left\lVert #1 \right\rVert}

\begin{document}

\title[Reeb Graph Reconstruction]{Faithful Reeb Graph Reconstruction of a Tectonic Subduction Zone from Earthquake Hypocenters}

\author{Halley Fritze}
\address{University of Oregon\\
    Eugene, OR, USA}
\email{hfritze@uoregon.edu}

\author{Sushovan Majhi}
\address{George Washington University\\
    Washington, DC, USA}
\email{s.majhi@gwu.edu}

\author{Marissa Masden}
\address{University of Puget Sound\\
    Tacoma, WA, USA}
\email{mmasden@pugetsound.edu}

\author{Atish Mitra}
\address{Montana Technological University\\
    Butte, MT, USA}
\email{amitra@mtech.edu}

\author{Michael Stickney}
\address{Montana Technological University\\
    Butte, MT, USA}
\email{mstickney@mtech.edu}


\begin{abstract}

An important problem in topological data analysis (TDA)---of both theoretical and practical interest---is to reconstruct the topology and geometry of an underlying (usually unknown) metric graph from possibly noisy data sampled around it.
Reeb graphs have recently been successfully employed in abstract metric graph reconstruction under  Gromov--Hausdorff noise: the sample is assumed to be metrically close to the ground truth. However, such a strong global density guarantee is often unavailable, making the existing Reeb graph-based methods unusable. 
A very different yet more relevant paradigm focuses on the reconstruction of metric graphs---embedded in the Euclidean space---from Euclidean samples that are only Hausdorff-close. We relax the density assumption to give provable geometric reconstruction schemes, even when the sample is metrically close only locally, but still provide provable guarantees for the successful geometric reconstruction of Euclidean graphs under the Hausdorff noise model. 

We apply our graph reconstruction techniques to reconstruct earthquake plate tectonic boundaries from the global earthquake catalog. The SLAB2.0 model is a comprehensive spatial summary of all known subduction zone slabs on Earth. We reconstruct parts of the SLAB2.0 model from possibly noisy earthquake hypocenter data.
\end{abstract}

\maketitle

\section{Introduction}
Modern technologies facilitate the easy and fast acquisition of large-scale geometric data. In many real-world applications, data come from a metric graph (known in bio sciences and astronomical sciences as filamentary structures). While some of these graphs are \emph{abstract}---e.g., sensor networks---other common applications involve samples around graphs embedded in some Euclidean space $\R^d$; we call such a graph \emph{embedded}.
Examples include GPS traces around road networks or maps \cite{Ahmed:2015:CTD:2820783.2820810}, earthquake locations around faults and plate tectonic boundaries \cite{geological-survey}, and more. The presence of noise in the data makes the reconstruction problem a computationally challenging problem in topological and geometric data analysis. 

Informally, a \emph{metric graph} $(G,d_G)$ is a combinatorial graph $(V,E)$ with the set of vertices $V$ and set of edges $E$, along with a \emph{length metric} $d_G$ that makes $G$ topologically equivalent to a $1$-dimensional simplicial complex; see \cite{BuragoBuragoIvanov} for a more formal definition. A sample around $G$ is commonly modeled as a finite metric space $(S, d_S)$. 
Then, a neighborhood graph $\mathcal{N}$ is computed at an appropriate scale as the input to the reconstruction algorithms. We choose to use the Vietoris--Rips graph $\mathcal{R}_\eta(S)$ on the raw data $S$ at a scale $\eta$, as defined in Section~\ref{sec:hausdorff}.

The reconstruction of metric graphs---both abstract and embedded---has been well-studied in the literature; see for example \cite{10.1145/1998196.1998203,chazal2013,dey_et_al:LIPIcs.SoCG.2018.31,NIPS2011_3a077244, Lecci:2014:SAM:2627435.2697074,Majhi2023,10.1145/3534678.3539029}.
Despite the existing reconstruction efforts for abstract graphs under the Gromov--Hausdorff setting, not many successful developments have surfaced to topologically and geometrically reconstruct embedded graphs from Euclidean samples that are only Hausdorff close to the ground truth. 
For embedded graphs, we particularly mention Majhi \cite{Majhi2023} for their technique using the Vietoris--Rips complexes. The resulting abstract simplicial complexes remain only topologically faithful, failing to produce a geometric graph as the \emph{output} in the same ambient space. This paper primarily aims to mitigate such caveats, developing provable reconstruction schemes for embedded graphs from noisy samples under the Hausdorff distance.

\paragraph{\bf Motivation and Contributions}
We are inspired by the recent pivotal success of Reeb graphs \cite{Chazal2015,chazal2013,NIPS2011_3a077244} in the reconstruction of abstract metric graphs geometrically.
The motivation for this work is two-fold.  (1) First, in the abstract setting, we relax the overly strong and impractical Gromov--Hausdorff assumption to provide reconstruction guarantees in Theorem~\ref{thm:geom-recon} from a locally close  $(\varepsilon, R)$-approximate (Def~\ref{def:eps-R-approx}) sample or neighborhood graph thereof. (2) Secondly, we consider the reconstruction of an embedded metric graph $G\subset\R^d$ such that the Euclidean distance induces the intrinsic distance $d_G$ on $G$ and the raw data $S\subset\R^d$ is only Hausdorff close to $G$ in $\R^d$, using the notion of restricted distortion (Def~\ref{def:rest-dist}) of a graph as a novel sampling parameter for embedded graphs - see Theorem~\ref{thm:main-thm}. In addition to our theoretical results, we demonstrate  our techniques in reconstructing earthquake plate tectonic boundaries from sensor data.

\section{Preliminaries}\label{sec:preli}

Let $G\subset\R^d$ be an embedded metric graph. $G$ comes equipped with the standard Euclidean metric $\|\cdot\|$.
Using the notion of the length $L(\gamma)$ of a continuous path $\gamma$ in $\R^d$, we can define the intrinsic or shortest path metric $d_G$ on $G$ as follows:
For any two points $x_1,x_2\in G$, it is the infimum of lengths of paths joining $x_1,x_2$. If $G$ is path-connected, $d_G$ defines a metric on $G$. For some basic concepts of metric geometry, see Appendix A.


\begin{definition}[$\alpha$-Reeb Graph]\label{def:alpha-Reeb}
Let $\alpha>0$, let $d(x)$ be the distance from a fixed base point, and let $\mathcal{I}=\left\{I_i\right\}$ be a covering of the range of $d$ by open intervals of length at most $\alpha$. The transitive closure of the relation $x \sim_\alpha y$ if and only if $d(x)=d(y)$ and $x, y$ are in the same path connected component of $d^{-1}\left(I_i\right)$ for some interval $I_i \in \mathcal{I}$ is an equivalence relation that is also denoted by $\sim_\alpha$. The quotient space $G_\alpha=X / \sim_\alpha$ is called the $\alpha$-Reeb graph (associated with the covering $\mathcal{I}$) of $d$ and we denote by $\pi: X \rightarrow G_\alpha$ the quotient map.  $\pi$ is continuous the function $d$ induces a function $d_*: G_\alpha \rightarrow \mathbb{R}_{+}$that satisfies $d=d_* \circ \pi$. See Figure~\ref{fig:alpha-reeb} in Appendix.
\end{definition}

\section{Hausdorff Noise Model}\label{sec:hausdorff}
Let $G\subset\R^d$ be an embedded, compact graph and $S\subset\R^d$ a finite sample
around $G$ such that their Hausdorff distance $d_H(S,G)$ is sufficiently small.
The sample and the graph come equipped with the standard Euclidean metric. The Euclidean distance on $G$ induces an intrinsic metric, denoted $d_G$, coming from the shortest paths on $G$.

\paragraph{\bf The Vietoris--Rips Graph}
For any compact $Y\subset\mathbb{R}^d$, the \emph{Vietoris--Rips graph} at scale $\eta>0$, denoted $\mathcal{R}_\eta(Y)$, is defined as a metric graph in the following manner. Combinatorially, $\mathcal{R}_\eta(Y)$  
is a graph whose vertex set is $Y$, and a pair $(y_1,y_2)\in Y$ forms an edge whenever $\|y_1-y_2\|\leq\eta$.
The combinatorial graph is then metrized in the following way: for each edge $e=(y_1,y_2)\in \mathcal{R}_\eta(Y)$, take an edge of length $\|y_1-y_2\|$ that is isometric to the real segment $[0,\|y_1-y_2\|]$. Then, define a metric on $\mathcal{R}_\eta(Y)$ by taking the shortest path metric, which we denote by $d^\eta_Y$. We note that $\mathcal{R}_\eta(Y)$ is $\Delta$-dense for any $\Delta\geq\eta$.  For $d^\eta_Y$  to be finite metric, the Euclidean thickening $Y^{\frac{\eta}{2}}$ of $Y$ must be path-connected. 


\paragraph{\bf Restricted Distortion}
Let $G\subset\R^d$ be an embedded metric graph. 
The restricted distortion of $G$, denoted $\delta^\eta_R(G)$, is parametrized by $\eta>0$ and $R>0$.

\begin{definition}[Restricted Distortion]\label{def:rest-dist}
For $\eta>0$ and $R>0$, the \emph{restricted distortion} or \emph{$(\eta,R)$-distortion} of an embedded metric graph $G\subset\R^d$ is defined as
\begin{equation*}\label{eq:delta-eps-R}
\delta^\eta_R(G)\coloneqq \sup_{d_G(x_1,x_2)\geq R} \frac{d_G(x_1,x_2)}{d^\eta_G(x_1,x_2)}.
\end{equation*}
\end{definition}
It immediately follows from the definition that $\delta^\eta_R(G)$ is a non-decreasing function of $\eta$ and a non-increasing function of $R$. 
Moreover, $\delta^\eta_R(G)\leq\delta(G)$, the global distortion of embedding of $G$. 
Indeed, we have
$\lim_{\substack{\eta\to\infty\\R\to0}}\delta^\eta_R(G)=\delta(G)$.

In case $G$ is compact, a remarkable property of the restricted distortion is that for any $R>0$, the restricted distortion $\delta^\eta_R(G)\to1$ as $\eta\to0$. The convergence provides the following important lemma. 

\begin{lemma}[$(\eps,R)$--approximation]\label{lem:main-approx}
Let $G\subset\R^d$ be an embedded graph. Let $R>0$ be a fixed number 
and $\varepsilon\leq\frac{1}{2}R$. 
If we choose
$\eta\in(0,R]$ small enough such that $\delta^{\eta}_{R}(G)\leq\frac{4}{3}$, then for any finite sample $S\subset\R^d$ with $d_H(G,S)<\frac{1}{4}\eta$, the metric space $(\Ri_\eta(S),d^\eta_S)$ is an $\left(\varepsilon,R\right)$-approximation of $(G,d_G)$.
\end{lemma}


\section{Reconstruction using Reeb Graphs}\label{sec:recon}
In this section, we present our main geometric reconstruction result. We begin by showing the following lemma, which generalizes the reconstruction results of \cite{chazal2013} under a weaker sampling assumption. Sketches of proofs are presented in the appendix.

\begin{lemma}\label{lem:diam-connected-component} 
Let $(X,d_X)$ and $(Y,d_Y)$ be compact geodesic metric spaces, and suppose there are $\eps, R >0$ such that  $X$ is  a $R$-dense $(\eps,R)$-approximation  of $Y$. Let base points $x_0 \in X$ and $y_0 \in Y$ be fixed and let $d_{x_0}(x)= d_X(x_0,x)$ and $d_{y_0}(y)= d_Y(y_0,y)$ be the distance functions from the base points, and let $M$ be the diameter of $Y$. For $d \ge \alpha \ge 0$,  let $D(X,d,\alpha)$ be the maximum diameter of any connected component of $d_{x_0}^{-1}([d-\alpha, d+\alpha])$, and let $D(Y,d,\alpha,M,R,\eps)$ be the maximum diameter of any connected component of $d_{y_0}^{-1}([d-\alpha - (\frac{4M}{R}+4)\eps, d+\alpha+ (\frac{4M}{R}+4)\eps])$. Then  $D(X,d,\alpha) \le D(Y,d,\alpha,M,R,\eps) + (\frac{4M}{R}+2)\eps$.
\end{lemma}

\begin{theorem}[$\alpha$-Reeb Reconstruction]\label{thm:geom-recon}
Let $(G, d_G)$ be a compact metric graph of diameter $M$ and $(X, d_X)$ a compact geodesic space such that $X$ is $R$-dense and an $(\eps,R)$- approximation of $G$. 
Let $r \in X$, $\alpha  > 0$, and $\mathcal{I}$ be a finite covering of the segment $[0, diam{(X)}]$ by open intervals of length at most $\alpha$ such that the $\alpha$-Reeb graph $G_{\alpha}$ associated to $\mathcal{I}$ and the function $d = d_X (r, .) :X\to \mathbb{R}$ is a finite graph. 
Then $d_{GH}(X, G_\alpha) \le (\beta_1(G_{\alpha}) + 1)[4(2 + N_{E,G}(4(\alpha + \frac{4\eps M}{R} +4\eps))](\alpha + \frac{4\eps M}{R} +4\eps) + (\frac{4M}{R}+2)\eps$, where $N_{E,G}(\ell)$ is the number of edges of $G$ of length at most $\ell$.
Moreover, as $R\to\infty$, this  Gromov--Hausdorff approximation reduces to $d_{GH}(X, G_\alpha)\leq (\beta_1(G_{\alpha}) + 1)[4(2 + N_{E,G}(4(\alpha +4\eps))](\alpha +4\eps) + 2\eps$.
\end{theorem}

\paragraph{\bf Main Reconstruction Results}
We now present our main result to reconstruct the geometry of an embedded graph under the Hausdorff noise.
\begin{theorem}[Reconstruction under Hausdorff noise]\label{thm:main-thm}
Let $G\subset\R^d$ be a connected, compact embedded graph of diameter $M$. Let $\varepsilon>0$ and $R\in[2\varepsilon,\infty]$.
Choose $\eta\in(0,R)$ such that $\delta^\eta_R(G)\leq\frac{4}{3}$.
Let $S\subset\R^d$ be a compact subset such that $d_H(S,G)  < \frac{1}{4}\eta$. Let $r \in \mathcal{R}_\eta(S)$, $\alpha  > 0$, and $\mathcal{I}$ be a finite covering of the segment $[0, Diam(\mathcal{R}_\eta(S)]$ by open intervals of length at most $\alpha$ such that the $\alpha$-Reeb graph $G_{\alpha}$ associated to $\mathcal{I}$ and the function $d = d_S (r, .) :\mathcal{R}_\eta(S))\to \mathbb{R}$ is a finite graph. Then,  
$d_{GH}(X, G_\alpha) \le (\beta_1(G_{\alpha}) + 1)[4(2 + N_{E,G}(4(\alpha + \frac{4\eps M}{R} +4\eps))](\alpha + \frac{4\eps M}{R} +4\eps) + (\frac{4M}{R}+2)\eps$, where $N_{E,G}(\ell)$ is the number of edges of $G$ of length at most $\ell$.
\end{theorem}


\section{Application: Earthquake Plate Tectonic Boundary Reconstruction}
Motivated by \cite{chazal2013}, we seek to apply our techniques of Reeb graph reconstruction to reconstruct tectonic slab contours \cite{slab} using earthquake hypocenters  \cite{geological-survey}; see \figref{data-and-truth}.
Whereas the $\alpha$-Reeb complex is guaranteed to reconstruct the topological structure of the graph with sufficiently small noise, an earthquake fault network is a geospatial object, and the graph reconstruction needs to not only capture the topological structure of the fault network but also follow the location of the faults. Earthquake location accuracy and uncertainty varies according to available seismographic data, thus, a local noise model is appropriate. We introduce a geometric post-processing in \ssecref{post-process}.

Currently, the $\alpha$-Reeb construction only uses hypocenters measured by latitude, longitude, and depth (\figref{data-and-truth}). By including additional features like location error ellipsoids (see Remark \ref{rem:ellipsoid}), we predict that our graph reconstruction will be more accurate.

\subsection{Comparison with SLAB2.0}\label{ssec:slab}

The SLAB2.0 model \cite{slab} is a comprehensive spatial analysis of all known subduction zone slabs on Earth in unprecedented detail, provided as contour lines, as in \figref{data-and-truth}. We use this as a benchmark for comparison for our geometric construction results. In \figref{results}, we reconstruct each contour line using earthquake hypocenters in 20 km depth ranges. Here, we choose our parameters $R=500$ and $\alpha=1000$, and we post-process our $\alpha$-Reeb graph to preserve the geometric features of the data. Due to the sparsity of our data set, we cannot achieve the exact lengths of the contour lines in \figref{data-and-truth}. However, when comparing a single reconstructed contour with the SLAB2.0 contour, we can see the preservation of the overall geometric features of the slab in our $\alpha$-Reeb graph. It should be kept in mind that our reconstruction uses only earthquake hypocenters, while SLAB2.0 uses additional relevant seismic data, too.

\section{Discussions}

We present provable reconstruction schemes for graphs embedded in Euclidean space from Hausdorff-close noisy samples. In our approach, we build a neighborhood graph on the sample that is metrically close to the original graph only locally and demonstrate that its $\alpha$-Reeb graph provides a geometrically faithful reconstruction. The success of our geometric reconstruction immediately evokes the question of topological reconstruction of embedded graphs using the same Reeb-based techniques as a future research direction. The nature of our research is ongoing, and we envision making our experimentation with earthquake data more mature and comprehensive. 
Our reconstruction scheme can benefit several other application areas where the ground truth is an embedded graph.


\clearpage
\appendix

\printbibliography

\section{Some Metric Geometry}

Let $(X,d_X)$ be a metric space. 
Let $A$ and $B$ be compact, non-empty subsets. 
In the following two most prevalent reconstruction settings, noise in the data is traditionally quantified by (a) \textbf{Abstract}: Gromov--Hausdorff distance (Def~\ref{def:gh}) between the abstract graph and a neighborhood graph $\mathcal{N}$ built on the raw data at an appropriate scale, and
(b) \textbf{Embedded}: Hausdorff distance (Def~\ref{def:dH}) in case $G$ is embedded in $\R^d$ and $S\subset\R^d$ is a finite Euclidean subset equipped with the standard Euclidean norm $\|\cdot\|$.
\begin{definition}[Hausdorff Distance]\label{def:dH}
The \emph{Hausdorff distance} between the $A$ and $B$, denoted $d^X_H(A, B)$, is defined as
\[
d^X_H(A, B) \coloneqq 
\max\left\{\sup_{a\in A}\inf_{b\in B}d_X(a,b),\sup_{b\in A}\inf_{a\in B}d_X(a,b)\right\}.
\]
In case $X\subset\mathbb{R}^N$ and $A,B,X$ are all equipped with the Euclidean metric, we simply write $d_H(A,B)$.
\end{definition}

\paragraph{The Gromov--Hausdorff Distance}
A \emph{correspondence} $\C$ between any two metric spaces $(X,d_X)$ and $(Y,d_Y)$ is defined as a subset of $X\times Y$ such that (a) for any $x\in X$, there exists $y\in Y$ such that $(x,y)\in\C$, and (b) for any $y\in Y$, there exists $x\in X$ such that $(x,y)\in\C$.
We denote the set of all correspondences between $X,Y$ by $\C(X,Y)$. The
\emph{distortion} of a correspondence $\C\in\C(X,Y)$ is defined as:
\[
\mathrm{dist}(\C)\coloneqq\sup_{(x_1,y_1),(x_2,y_2)\in\C}
|d_X(x_1,x_2)-d_Y(y_1,y_2)|.
\]
\begin{definition}[Gromov--Hausdorff Distance]\label{def:gh} Let $(X,d_X)$ and $(Y,d_Y)$ be two compact metric spaces. The \emph{Gromov--Hausdorff} distance between $X$ and $Y$, denoted by $d_{GH}(X,Y)$, is defined as:
\[
d_{GH}(X,Y)\coloneqq
\frac{1}{2}\left[\inf_{\C\in\C(X,Y)}\mathrm{dist}(\C)\right].
\]
\end{definition}
A slightly weaker and more local notion is $(\eps,R)$- approximation for constants $\eps,R>0$. 
The notion has already been used in \cite{10.1145/1998196.1998203,MajhiStability} in the context of reconstruction.
However, we use a slightly different version to match the traditional definition of the Gromov--Hausdorff distance.
\begin{definition}[$(\eps,R)$-Approximation]\label{def:eps-R-approx}
 metric For $\eps,R$ positive, two metric spaces $(X,d_X)$ and $(Y,d_Y)$ are said to be an $(\eps,R)$-approx\-imation of each other if there exists a correspondence $\C\subset X\times Y$ such that for any $(x_1,y_1),(x_2,y_2)\in\C$
with $$\min\{d_X(x_1,x_2), d_Y(y_1,y_2)\}\leq R,$$ we have $|d_X(x_1,x_2)-d_Y(y_1,y_2)|\leq 2\eps$.
\end{definition}
To see that this is a weaker notion, we note that $d_{GH}(X, Y)<\eps$ implies that the two metric spaces are $(\eps, R)$ for any $R>0$. However, an $(\eps, R)$-approximation may not always imply $\eps$-closeness in the Gromov--Hausdorff distance; see \cite[Example 2.7]{MajhiStability}.

\begin{definition}[$\Delta$-Dense]\label{def:R-dense}
A path metric space $(X,d_X)$ is said to be $\Delta$-dense if for 
 $x,x' \in X$ there exists a sequence $\{x_0=x, x_1, \cdots, x_n=x'\}$ such that for all $i=0, \cdots, n-1$, we have $d_X(x_i,x_{i+1}) \le \Delta$ and $d_X(x,x')=\sum_{i=0}^{n-1} d_X(x_i,x_{i+1})$.
\end{definition}

 \section{Additional Results and  Proofs}

Here we present sketches of proofs of the results in section \ref{sec:recon}.
 
\begin{proof}[Proof of Lemma \ref{lem:diam-connected-component}]
Let  $\mathcal{C} \subset X \times Y$ be a $(\eps, R)$ correspondence. As $X$ is $R$-dense, 
distances in $X$ are distorted from corresponding distances in $Y$ by at most $(\frac{4M}{R}+2)\eps$ (see \cite{10.1145/1998196.1998203}). Given points $x,x'$ and a continuous path $\gamma_X$ joining them in a connected component of $d_{x_0}^{-1}([d-\alpha, d+\alpha])$, and letting $(x,y),(x',y') \in \mathcal{C}$, we can construct (by subdividing the path $\gamma_X$ in $R$-short subpaths) a continuous path  $\gamma_Y$ joining $y,y'$ in a connected component of $d_{y_0}^{-1}([d-\alpha - (\frac{4M}{R}+4)\eps, d+\alpha+ (\frac{4M}{R}+4)\eps])$. The result then follows.

\end{proof}

\begin{proof}[Proof of Theorem \ref{thm:geom-recon}]
Proposition 3 of \cite{Chazal2015} gives a bound of the diameters of the connected components of the $\alpha$-thick annuli in terms of $N_{E,G}(\alpha)$ - the number of edges of length at most $\alpha$. That bound, alongwith Lemma \ref{lem:diam-connected-component},  gives the   Gromov-Hausdorff bound for $\alpha$-Reeb graphs. Here $\beta_1(G_{\alpha})$ is the first Betti number of $G_{\alpha}$.
\end{proof}

\begin{proof}[Proof of Theorem \ref{thm:main-thm}]
From Lemma~\ref{lem:main-approx}, we have that $X=\mathcal{R}_\eta(S)$ is an $(\varepsilon, R)$-approximation of $(G,d_G)$. As we also noted, $X$ is $R$-dense. Theorem~\ref{thm:geom-recon} then implies that
$d_{GH}(X, G_\alpha) \le (\beta_1(G_{\alpha}) + 1)[4(2 + N_{E,G}(4(\alpha + \frac{4\eps M}{R} +4\eps))](\alpha + \frac{4\eps M}{R} +4\eps) + (\frac{4M}{R}+2)\eps$.
If the smallest edge length of $G$ is at least $4(\alpha+\frac{4\varepsilon M}{R}+4\varepsilon)$, then the Gromov--Hausdorff distance can be bounded by $(\beta_1(G_\alpha)+1)(8\alpha+ \frac{36 \eps M}{R} +34\eps)$.
\end{proof}

\begin{remark}\label{rem:ellipsoid}
Earthquake epicenter data is often  available with error ellipsoids. We can incorporate the bounds of the axes of the error ellipsoids into Theorem \ref{thm:main-thm}, getting  more technically complicated versions of the Gromov-Hausdorff bound. We omit those bounds here for brevity.
\end{remark}

\section{Figures and Post Processing Techniques}

\begin{figure}[!ht]
    \centering
    \includegraphics[width=0.7\linewidth]{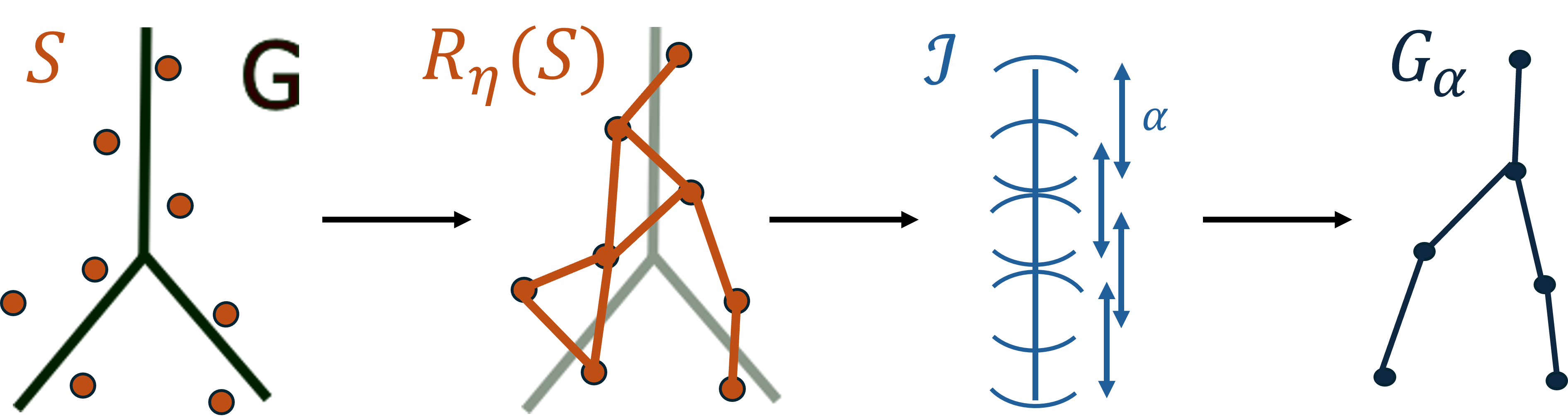}
    \caption{The $\alpha$-Reeb reconstruction under a Hausdorff noise model. First consider a Hausdorff close sample $S$ of a graph $G$. Second, construct a Vietoris-Rips graph $R_\eta(S)$ on the sample
    for some $\eta >0$. 
    Choose a random base point and filter by distance to the base point. 
    Use this to construct the $\alpha$-Reeb graph $G_\alpha$.}
    \label{fig:alpha-reeb}
\end{figure}

\begin{figure}
    \centering
    \includegraphics[width=0.3\linewidth]{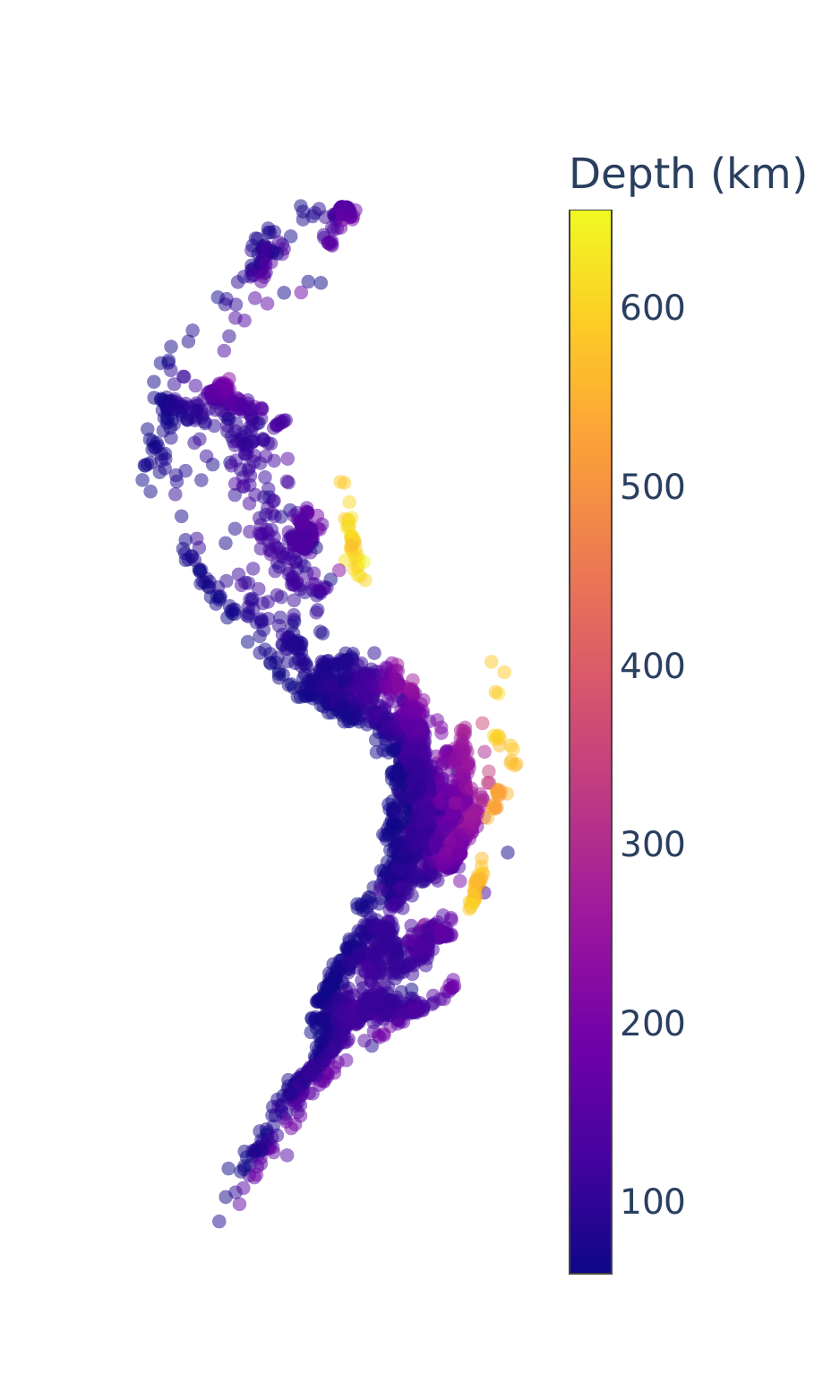}
    \includegraphics[width=0.3\linewidth]{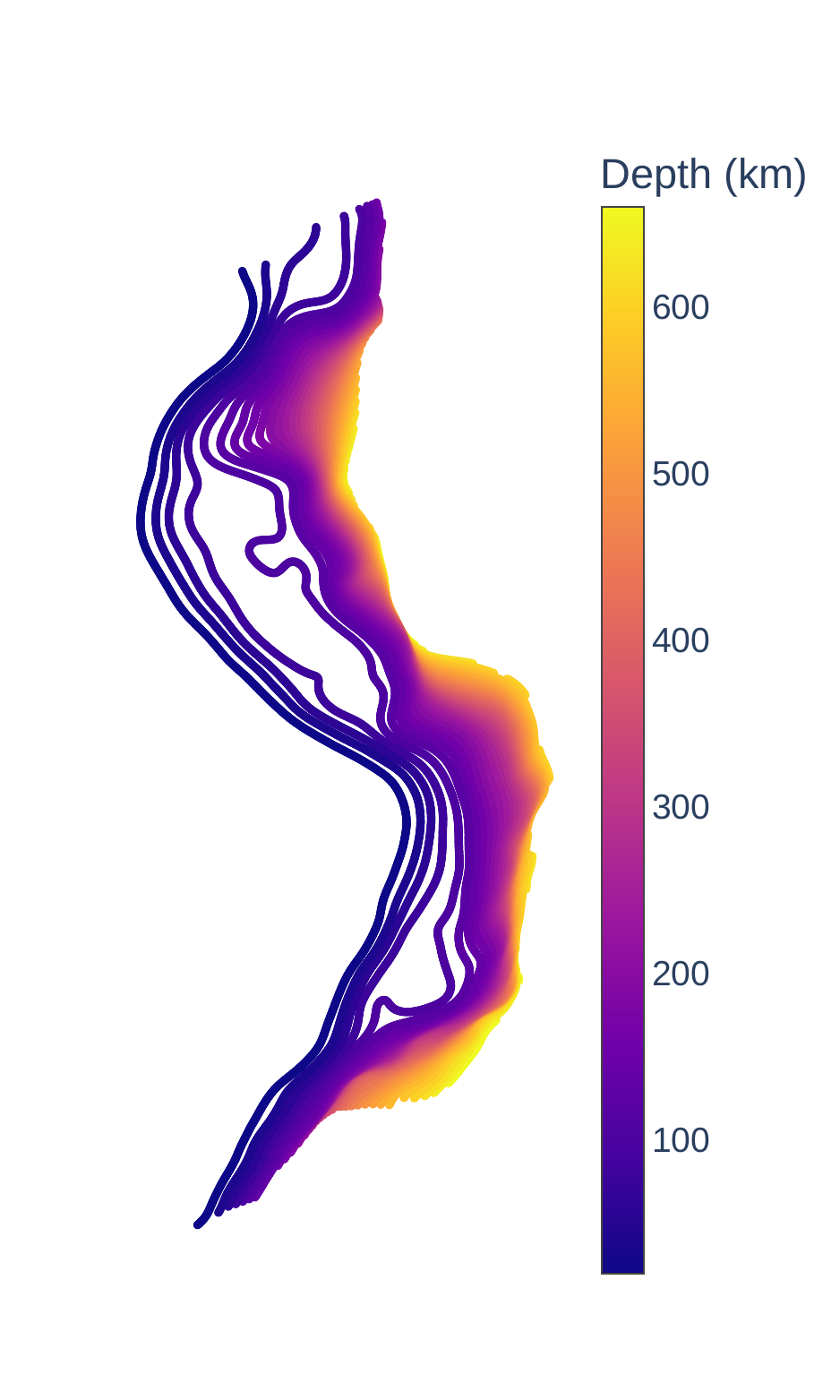}
    \caption{Left: Earthquake hypocenters in western South America with focal depths greater than or equal to 60 km and magnitudes larger than 4. 
    Right: Subduction slab contours from SLAB2.0 of the same region.}
    \label{fig:data-and-truth}
\end{figure}
 
\begin{figure}
    \centering
    \includegraphics[width=0.3\linewidth]{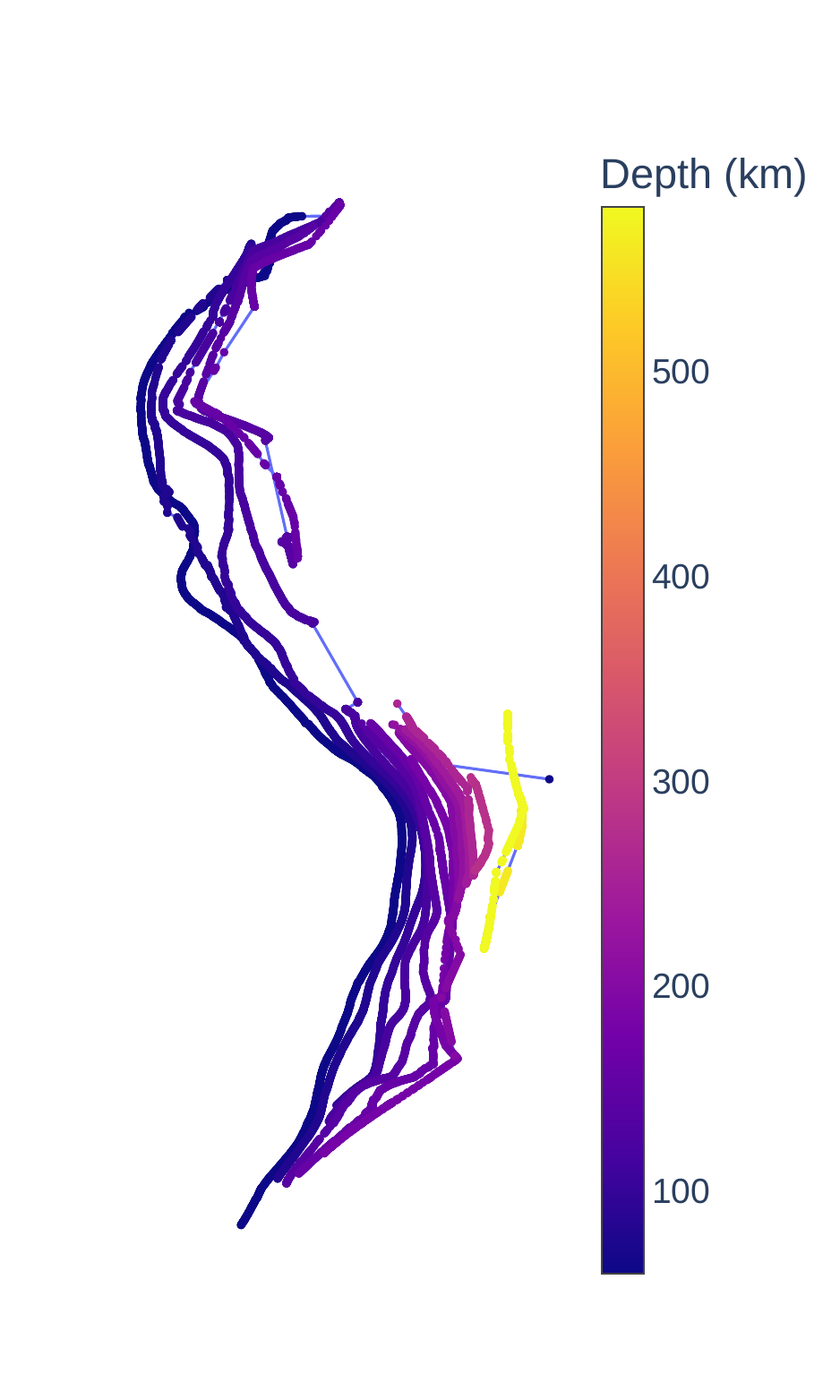}
    \includegraphics[width=0.3\linewidth]{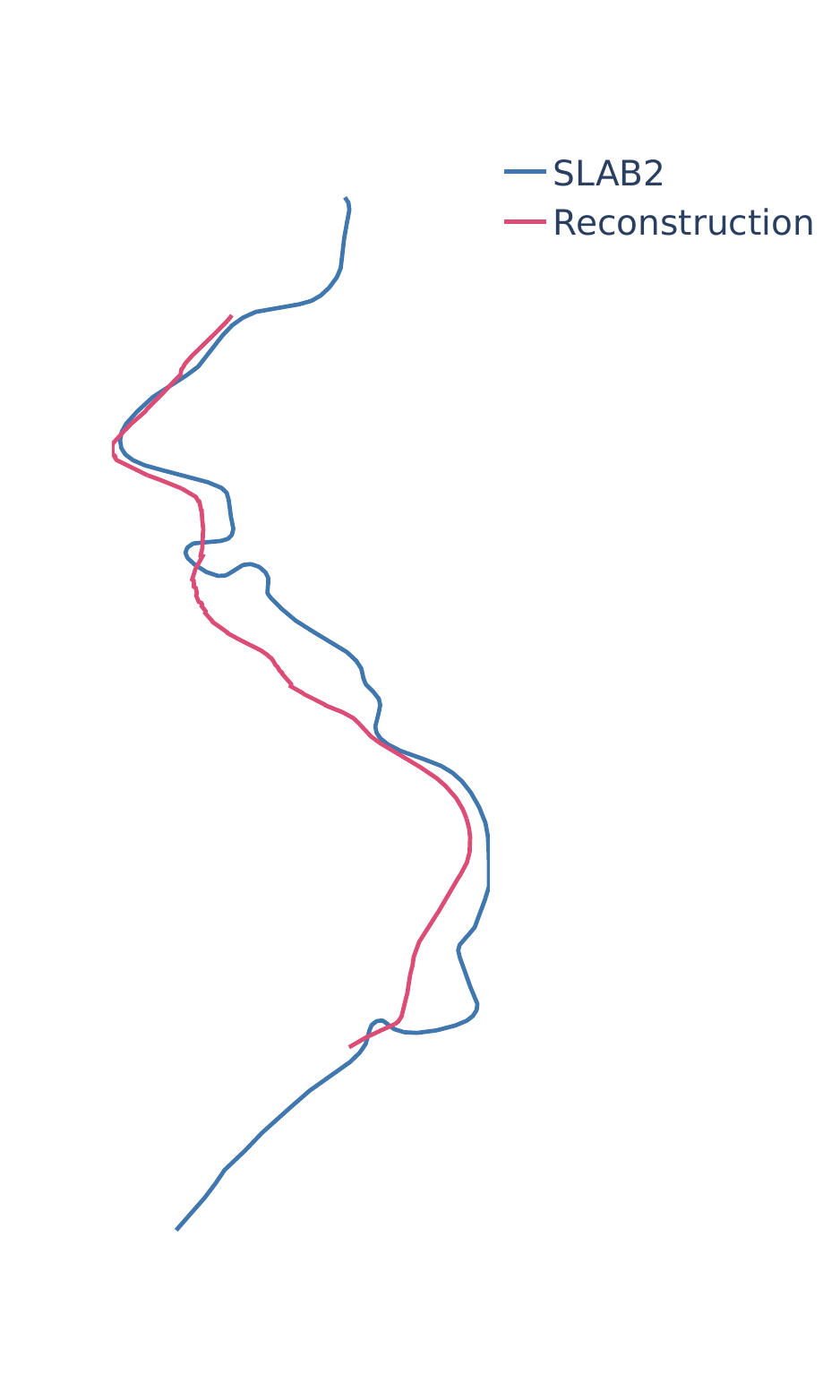}
    \caption{Left: Geometric graph reconstruction contour lines of the western South America tectonic slab using earthquake hypocenters in 20 km depth ranges beginning at 60 km.
    Right: Geometric graph reconstruction using earthquake hypocenters with a depth range of 90--110 km, compared to the 100 km depth contour of SLAB2.0.}
    \label{fig:results}
\end{figure}

\subsection{Geometric Smoothing (Post Processing)}\label{ssec:post-process}

 The $\alpha$-Reeb graph provides a topologically correct reconstruction of a given graph under the given assumptions (\thmref{main-thm}). To geometrically embed the nodes and edges of this graph in a way that respects  the underlying geometry of the data, for each edge in the graph, we obtain the list of data points from the underlying data that are contained in the corresponding vertices of the $\alpha$-Reeb graph, 
 $\{p_i\}_{i \in v_1 \cup v_2}$ as well as their corresponding filter value $f(p_i)$. Each vertex is itself assigned a filter value, denoted $f(v)$, given by the median of the filter values in that vertex. Then, parameterizing the edge on the interval $[f(v_1), f(v_2)]$, for each $t$ in this interval we produce the embedding via a weighted average of locations, $E(t)=\sum p_i\cdot  w_i(t)$, where $w_i(t)$ are the normalized weights over the set of data contained in $v_1 \cup v_2$ given by a Gaussian kernel centered at $t$, obtained via $$w_i(t) \propto \exp\left(\frac{-(f(p_i)-t)^2}{2\sigma^2}\right).$$ 

 Selecting $\sigma < \frac{\alpha}{4}$ ensures approximate smoothness of the embedding between adjacent edges. 

\end{document}